\documentclass[manuscript,nonacm]{acmart}

\usepackage{fancyhdr}
\usepackage{appendix}
\usepackage{listings}
\usepackage{pifont}
\usepackage{multicol}
\usepackage{multirow}
\usepackage{booktabs}
\usepackage[ruled,vlined]{algorithm2e}
\usepackage{svg}
\usepackage{comment}

\newcommand{\ignore}[1]{}

\newcommand{\cQ}{\mathcal Q}
\newcommand{\cF}{\mathcal F}
\newcommand{\cM}{\mathcal M}
\newcommand{\cS}{\mathcal S}
\newcommand{\cG}{\mathcal G}
\newcommand{\cT}{\mathcal T}

\SetKwComment{Comment}{/* }{ */}

\newcommand{\myparagraph}[1]{\vspace{3.5pt}\noindent \textbf{#1}}

\title{How to Tame Multiple Spending in Decentralized Cryptocurrencies}

\author{João Paulo Bezerra}
\email{joaopaulo.bezerra@telecom-paris.fr}
\affiliation{%
  \institution{LTCI, Télécom Paris, Institut Polytechnique de Paris}
  \country{France}
}

\author{Petr Kuznetsov}
\email{petr.kuznetsov@telecom-paris.fr}
\affiliation{%
  \institution{LTCI, Télécom Paris, Institut Polytechnique de Paris}
  \country{France}
}

\renewcommand{\shortauthors}{J. Paulo Bezerra, P. Kuznetsov}

\ccsdesc[500]{Theory of computation~Design and analysis of algorithms~Distributed algorithms}

\keywords{Quorum systems, decentralized trust, consistency measure, asset transfer, accountability}

\begin{document}




\begin{abstract}
The last decade has seen a variety of Asset-Transfer systems designed for decentralized environments.
To address the problem of \emph{double-spending}, these systems inherently make strong model assumptions and spend a lot of resources.
In this paper, we take a non-orthodox approach to the double-spending problem that might suit better realistic environments in which these systems are to be deployed.
We consider the \emph{decentralized trust} setting, where each user may independently choose who to trust by forming its local quorums.
In this setting, we define \emph{$k$-Spending Asset Transfer}, a relaxed version of asset transfer which bounds the number of times the same asset can be spent.
We establish a precise relationship between the decentralized trust assumptions and $k$, the optimal \emph{spending number} of the system.
\end{abstract}

\maketitle

The conference version of this paper is available at~\cite{bacrypto2022}.

\begin{acks}
This work was supported by TrustShare Innovation Chair.
\end{acks}

\section{Introduction}

\myparagraph{Fault models and quorum systems.}
Distributed protocols, such as consensus and broadcast, are generally built to be resilient against arbitrary (Byzantine) faults of system members.
To maintain consistency and liveness, these protocols typically have to assume that only a fraction of system members are Byzantine.
In the special case of an \emph{uniform} fault model, where faults of system members are identically and independently distributed,  bounds on the number $f$ of Byzantine members that can be tolerated are well known: less than half of system members ($f<n/2$) in synchronous networks (using digital signatures)~\cite{katz2006expected}, and less than one third ($f<n/3$) in asynchronous or partially synchronous networks~\cite{bracha1983resilient}.

More general fault models can be captured via \textit{quorum systems}~\cite{malkhi1998byzantine,vukolic2013origin}, collections of subsets of system participants, called \emph{quorums}, that meet two conditions: in every system run, (i)~\emph{every} two quorums should have at least one correct participant in common and (ii)~\emph{some} quorum should only contain correct participants.
Intuitively, quorums encapsulate \emph{trust} the system members express to each other. 
Every quorum can act on behalf of the whole system: an update or a query on the data is considered safe if it involves a (sufficiently trusted) quorum of replicas.

\myparagraph{Decentralized quorums.}
Conventionally, quorum assumptions are centralized: all participants share the same quorum system.
In some large-scale distributed systems, it might be, however, difficult to expect that all participants come to the same trust assumptions.
Recently, the quorum-based approach to system design has been explored in completely new way.
It started with system implementations~\cite{schwartz2014ripple,mazieres2015stellar} that allowed their users to not necessarily hold the same assumptions of who to trust, i.e., to maintain \emph{local} quorum systems.
Based on its local knowledge, a system member might have its own idea about which subsets of other participants are trustworthy and which are not.
We come, therefore, to the model of \emph{decentralized quorums}: each system member maintains its own quorum system. 

Great effort have been invested into improving protocols designed for uniform fault models~\cite{pbft,zyzzyva,sbft,hotstuff}, or in understanding which conditions on individual quorum systems are necessary and sufficient, so that some well-defined subset of participants can solve a problem \cite{garcia2018federated,cachin2020asymmetric,losa2019stellar}.
However, little is understood about the “damage” that Byzantine processes might cause if these conditions do not hold, e.g., in the decentralized quorum system model.
Intuitively, the more Byzantine processes there are or more strategically they are located in decentralized quorums, the more they can affect the system's consistency. 
But what exactly does ``more'' mean here? 

\myparagraph{Asynchronous cryptocurrencies.}
In this paper, we study this question on the example of asset-transfer systems (or \emph{cryptocurrencies}). 
Conventionally, the major challenge addressed by a cryptocurrency is to prevent \emph{double spending}, when a malicious or misconfigured user manages to spend the same coin more than once.
As was originally claimed by Nakamoto~\cite{nakamoto2008bitcoin}, preventing  double spending in systems with mutual distrust requires honest users to agree on the order in which the transactions must be executed, i.e., to solve the fundamental problem of \emph{consensus}~\cite{fischer1985impossibility}.
Bitcoin achieves probabilistic \emph{permissionless} consensus assuming a  synchronous system and using the proof-of-work mechanism. 
The protocol is notoriously energy-consuming and slow.
Since then, a long line of systems used consensus for implementing cryptocurrencies in both permissionless and permissioned contexts.

It has been later observed that cryptocurrencies do not always require consensus in general~\cite{gupta2016non,cons-crypto}.
It turns out that it is not always necessary to maintain a totally ordered set of transactions, a specific  partial order may suffice. 
Intuitively, if we assume that each account has a single dedicated owner, it is sufficient to agree on the order of outgoing  transactions \emph{per account}.
Transactions operating on different accounts can be ordered arbitrarily without affecting correctness. 
Double spending is excluded, as no user can publish ``conflicting'' transactions on its account (spending more money than its account holds).
Recently proposed asynchronous (\emph{consensus-free}) cryptocurrencies~\cite{fastpay,astro-dsn}
exhibit significant advantages over consensus-based protocols in terms of scalability, performance and  robustness.
However, as they still rely on classical quorum systems, they are hardly applicable at a large scale.

\myparagraph{Contributions.}
In this paper, we explore the potential of decentralized quorums in implementing asynchronous cryptocurrencies.
Naturally, this model allows us to formally capture the \emph{double spending} phenomenon. 
In a way, we mimic the principle followed by real-world financial systems, where double spending is not an enemy but a routine phenomenon.

We introduce \emph{$k$-Spending Asset Transfer}, a relaxed cryptocurrency abstraction suitable for decentralized trust models. 
Notice that in this model, quorums chosen by correct processes might not be globally consistent, i.e., may not all intersect in a correct process.
Byzantine processes can exploit this by enforcing correct processes to accept conflicting transactions with the same input,
resulting in \emph{multiple spending}.

Intuitively, a $k$-spending asset-transfer system guarantees that an asset can be spent at most $k$ times.
The adversary may exploit the fact that quorums chosen by correct processes might not have a correct overlap, and make them accept conflicting transactions.
However, any instance of multiple-spending that affects correct participants should be eventually detected and a proof of misbehaviour against the Byzantine process should be published.   

As a bold analogy, one can think of a global financial trading system, where every national economy benefits from mutual trust, while cross-border interactions are less reliable.
But if the lack of trust is exploited by a cheating trader, correct participants should eventually be able to detect and punish the cheater, e.g., by excluding from the system.

We show how the parameter $k$ in $k$-spending asset transfer relates to the structure of the underlying quorum assumptions.
We visualize these assumptions via a family $\cG_{\cQ,\cF}$ of graphs, one for each possible faulty set $F\in \cF$ and each quorum map $S$, mapping each process $p$ to an element in its local quorum system $\cQ(p)$.     
It turns out that the optimal number of times a coin can be spent in this system is precisely the maximal \emph{independence number} over graphs in $\cG_{\cQ,\cF}$. 

Thus, our contributions are three-fold. 
We introduce the abstraction of $k$-spending asset transfer that defines a precise bound $k$ on the number of times a given asset can be spent.
We represent decentralized trust assumption in the form of a family of trust graphs and show that its maximum independence number gives a lower bound on $k$. 
We present a $k$-asset transfer implementation that shows that the bound is tight.
In addition, the algorithm maintains an accountability mechanism that keeps track of multiple spending and publishes evidences of misbehavior.  

%
%
%


\myparagraph{Road map.}
The rest of the paper is organized as follows.
In Section~\ref{sec:model} we present our system model. 
Section~\ref{sec:dc_trust} introduces a graph representation of trust, used later in the paper to prove lower bounds in the protocols.
In Section~\ref{sec:transfer_system} we give the specification of \emph{$k$-Spending Asset Transfer} ($k$-SAT) abstraction and present a protocol for implementing it.
We show that our $k$-SAT algorithm is optimal in Section~\ref{sec:lower_bounds}, by relating it to a relaxed broadcast abstraction: \emph{$k$-Consistent Broadcast} ($k$-CB).
We overview related work in Section~\ref{sec:related_work}. Finally, we discuss the results and future work in Section~\ref{sec:conclusion}.

\section{System Model}
\label{sec:model}

\myparagraph{Processes.}
A system is composed of a set of \emph{processes} $\Pi = \{p_1,...,p_n\}$. 
Every process is assigned an \emph{algorithm} (we also say \emph{protocol}), an automaton defined as a set of possible \textit{states} (including the \textit{initial state}), a set of \textit{events} it can produce and a transition function that maps each state to a corresponding new state. 
An event is either an $input$ (a call operation from the application or a message received from another process) or an $output$ (a response to an application call or a message sent to another process); \textit{send} and \textit{receive} denote events involving communication between processes.

\myparagraph{Executions and failures.}
A \textit{configuration} $C$ is a collection of states of all processes. In addition, $C^0$ is used to denote a special configuration where processes are in their initial states.
An \textit{execution} (or a \textit{run}) $\Sigma$ is a sequence of events,
where every event is associated with a distinct process and
every \textit{receive}($m$) event has a preceding matching \textit{send}($m$) event. 
A process \textit{misbehaves} in a run (we also call it \emph{Byzantine}) if it produces an event that is not prescribed by the assigned protocol, given the preceding sequence of events, starting from the initial configuration $C^0$. 
If a process does not misbehave, we call it \emph{benign}.
In an infinite run, a process \textit{crashes} if it prematurely stops producing events required by the protocol; 
if a process is benign and never crashes we call it \emph{correct}, and it is 
\emph{faulty} otherwise. 
Let $\textit{part}(\Sigma)$ denote the set of processes that produce events in an execution $\Sigma$.

\myparagraph{Channels.}
Every pair of processes communicate over a \textit{reliable channel}: in every infinite run, if a correct process $p$ sends a message $m$ to a correct process $q$, $m$ eventually arrives, and $q$ receives a message from $p$ only if $p$ sent it.
We impose no synchrony assumptions. 
In particular, we assume no bounds on the time required to convey a message from one correct process to another. 
\myparagraph{Digital signatures.} We use asymmetric cryptographic tools: a pair public-key/private-key is associated with every process in $\Pi$~\cite{cachin2011introduction}.  
The private-key remains secret to its owner and can be used to produce a \emph{signature} for a statement, %
while the public-key is known by all processes and is used to \emph{verify} that a signature is valid. 
Every process have access to operations \textit{sign} and \textit{verify}: \textit{sign} takes the process' identifier and a bit string as parameters and returns a signature, while \textit{verify} takes the process' identifier, a bit string and a signature as parameters and return $b \in \{TRUE,FALSE\}$.
We assume a computationally bound adversary: no process can forge the signature for a statement of a benign process.

\myparagraph{Trust assumptions.}
We now define our decentralized trust model.  
A \emph{quorum system map} $\cQ: \Pi \rightarrow 2^{2^{\Pi}}$ provides every process with a set of process subsets: 
for every process $p$, $\cQ(p)$ is the set of \emph{quorums of $p$}.  
We assume that $p$ includes itself in each of its quorums: $\forall Q\in \cQ(p): p\in Q$.
Intuitively, $\cQ(p)$ consists of sets of processes $p$ expects to appear correct in system runs. 
From $p$'s perspective, for every quorum  $Q\in\cQ(p)$, there must be an execution in which $Q$ is precisely the set of correct processes. 
However, these expectations may be violated by the environment.
We therefore introduce a \emph{fault model} $\cF\subseteq 2^{\Pi}$ (sometimes also called an \emph{adversary structure}) stipulating which process subsets can be faulty.   
We assume \emph{inclusion-closed} fault models that, intuitively, do not force processes to fail: $\forall F\in \cF,\; F'\subseteq F: \; F'\in \cF$. 
From now on, we consider only executions $\Sigma$ that \emph{complies with $\cF$}, i.e., the set of faulty processes in $\Sigma$ is in $\cF$. 

Given a faulty set $F\in \cF$, a process $p$ is called \emph{live in $F$} if it has a \emph{live quorum in $F$}, i.e., $\exists Q\in \cQ(p): Q \cap F = \emptyset$.
%
%
For example, let the uniform \emph{$f$-resilient} fault model: $\cF=\{F\subseteq \Pi: |F|\leq f\}$.
If $\cQ(p)$ includes all sets of $n-f$ processes, then $p$ is guaranteed to have at least one live quorum in every execution.
On the other hand, if $\cQ(p)$ expects that a selected process $q$ is always correct ($q\in \cap_{Q\in\cQ(p)} Q$), then $p$ is not live in any execution with a faulty set such that $q\in F$. %

In the rest of the paper, we consider a \emph{trust model} $(\cQ,\cF)$, where $\cQ$ is a quorum map and $\cF$ is a fault model.

\section{Graph Representation of Trust}
\label{sec:dc_trust}

We use undirected graphs to depict possible scenarios of executions with trust assumptions $(\cQ,\cF)$.
Intuitively, each graph represents a situation where a correct process hears from a quorum before accepting a statement in a protocol.
Let $S: \Pi \rightarrow 2^{\Pi}, S(p) \in \cQ(p)$, be a map providing each process with one of its quorums,
and $\cS$ be the family of all possible such maps.
For a fixed faulty set $F \in \cF$ and $S \in \cS$, the graph $G_{F,S}$ is a tuple $(\Pi_{F},E_{F,S})$ defined as follows:

\begin{itemize}
    \item $\Pi_{F} = \Pi - F$, i.e.,  the set of correct processes;
    \item Nodes $p$ and $q$ are connected with an edge \textit{iff} their quorums $S(p)$ and $S(q)$ intersect in a correct process, i.e., $(p, q) \in E_{F,S} \Leftrightarrow S(p) \cap S(q) \not\subseteq F$.
\end{itemize}

\begin{example}
\label{ex:graphs}
Consider a system of four processes, where $\Pi = \{p_1,p_2,p_3,p_4\}$, $\cF = \{\{p_3\}\}$, and the individual quorum systems are:
\[ \cQ(p1) = \{\{p_1,p_2,p_3\}\} \ \ \ \ \   \cQ(p2) = \{\{p_1,p_2\},\{p_2,p_4\}\} \]
\[ \cQ(p3) = \{\{p_1,p_2,p_4\}\} \ \ \ \ \   \cQ(p4) = \{\{p_2,p_4\},\{p_3,p_4\}\} \]

Consider an execution with $F = \{p_3\}$, the set of correct processes $\Pi_F$ is $\{p_1,p_2,p_4\}$.
Let $S_1 \in \cS$ be a quorum map for $\cQ$ such that $S_1(p_1) = \{p_1,p_2,p_3\}$, $S_1(p_2) = \{p_1,p_2\}$ and $S_1(p_4) = \{p_2,p_4\}$.
The quorums of every pair of correct processes in the resulting graph $G_{F,S_1}$ intersect in $p_2 \in \Pi_F$, thus resulting in a fully connected graph.
Now let $S_2 \in \cS$ be another quorum map such that $S_2(p_1) = \{p_1,p_2,p_3\}$, $S_2(p_2) = \{p_2,p_4\}$ and $S_2(p_4) = \{p_3,p_4\}$.
Since $S_2(p_1) \cap S_2(p_4) \subseteq F$, the resulting graph $G_{F,S_2}$ has a missing edge.
Figure \ref{ex1_graphs} depicts $G_{F,S_1}$ and $G_{F,S_2}$.
\end{example}

\begin{figure}
\centering
\includegraphics[scale=0.6]{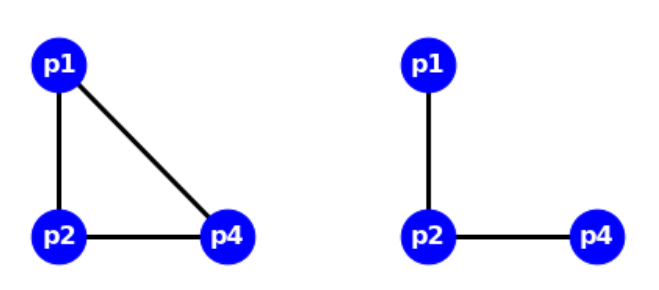}
\caption{Graph structures of Example~\ref{ex:graphs}: $G_{F,S_1}$ and $G_{F,S_2}$ respectively.}
\label{ex1_graphs}
\end{figure}

\myparagraph{Inconsistency number.}
We recall two useful definitions from graph theory: \emph{Independent Set} and \emph{Independence Number}.

\begin{definition} [Independent Set]

A set $C \subseteq V$ is an \emph{independent set} of $G = (V,E)$ iff no pair of nodes in $C$ is adjacent,
i.e., $\forall p,q \in C: (p,q) \notin E$.
$C$ is \emph{maximum} iff for every independent set C' of $G$: $|C'| \leq |C|$.

\end{definition}

\begin{definition} [Independence Number]

The \emph{independence number} of $G$ is the size of its maximum independent set(s).

\end{definition}

Given  the pair $(\cQ,\cF)$, 
we note $\cG_{\cQ,\cF}$ the family of graphs including all possible $G_{F,S}$, where $F \in \mathcal{F}$ and $S \in \cS$.

\begin{definition}[Inconsistency Number]
The \emph{inconsistency number} of $(\cQ,\cF)$, noted $\lambda(\cG_{\cQ,\cF})$, is the highest independence number among all $G_{F,S} \in \cG_{\cQ,\cF}$.
Formally,
Let $\mu:\cG_{\cQ,\cF} \rightarrow N$ map each $G_{F,S} \in \cG_{\cQ,\cF}$ to its independence number,
then $\lambda(\cG_{\cQ,\cF}) = max(\{\mu(G_{F,S})|G_{F,S} \in \cG_{\cQ,\cF}\})$.
\end{definition}

\begin{example}
\label{ex:indepedence}
Coming back to Example~\ref{ex:graphs}, the graph $G_{F,S_1}$ is fully connected, thus it has independence number $1$. 
On the other hand, the maximum independent set in $G_{F,S_2}$ is $\{p1,p4\}$, as a result, $G_{F,S_2}$ has independence number $2$.
Now consider a pair $(\cQ,\cF)$ where $\cQ$ and $\cF$ are the same as in Example~\ref{ex:graphs} (with this assumption only $p_3$ may fail in any execution).
Since $\forall S \in \cS,\forall F \in \cF: S(p_1) \cap S(p_2) \not\subseteq F$, it is easy to see that no graph has independence number higher than $2$ in $\cG_{\cQ,\cF}$, thus the inconsistency number of $(\cQ,\cF)$ is $2$.
\end{example}

\myparagraph{Computing inconsistency parameters.} %
A straightforward approach to find the inconsistency number of $(\cQ,\cF)$ consists in computing the independence number of all graphs $G_{F,S} \in \cG_{\cQ,\cF}$. 
The problem of finding the \textit{maximum independent set} of a graph, and consequently its independence number, is the \textit{maximum independent set problem} \cite{tarjan1977finding},
known to be \textit{NP-complete} \cite{miller2013complexity}.
Also, the number of graphs in $\cG_{\cQ,\cF}$ may exponentially grow with the number of processes.
However, as the graphs might have similar structures (for example, the same quorums for some processes may appear in multiple graphs),
in many practical scenarios, we should be able to avoid redundant calculations and reduce the overall computational costs, as we show for the uniform model.

\myparagraph{Inconsistency in the uniform model.} Centralized quorum systems generate graphs that are similar in structure and are therefore easier to analyse. Given a uniform quorum system $\cQ_{u}$,
we show how to calculate the spending number of $(\cQ_{u},\cF_{u})$, where $F_{u}$ includes every subset of processes with size $\leq f$.

\begin{theorem}
\label{th:classic_inc}
Let $(\cQ_{u},\cF_{u})$ be a uniform quorum system  with $n$ processes, where quorums have size $q$ and at most $f$ processes might fail. The inconsistency number of $(\cQ_{u},\cF_{u})$ is $\lfloor(\frac{n-f}{q-f})\rfloor$.
\end{theorem}

\begin{proof}
Fix any $F \in \cF_{u}$ of size $f$.
Let $G_{F,S} \in \cG_{\cQ_{u},\cF_{u}}$ be a graph whose independence number is the highest, and let $C_{max} = \{p_1,...,p_m\}$ be a maximum independent set in $G_{F,S}$.
Let $cor(Q)$ denote the number of correct processes in a quorum $Q$ and let $Q_i = S(p_i)$. It follows that $cor(Q_1)+...+cor(Q_m) \leq n-f$, since the quorums $Q_1,...,Q_m$ have no correct processes in common.
We can build a graph $G'_{F,S'}$ with an independent set $C' = \{p_1,...,p_m\}$ of the same size as $C_{max}$, where $\forall p \in C': F \subseteq S'(p)$, that is, the quorum for every $p \in C'$ includes all faulty processes.
It suffices to choose $S'(p_i)$ with any $q-f$ correct processes from $S(p_i)$, in addition to the $f$ faulty processes.
Now let $Q'_i \in S'(p_i)$. Since $cor(Q'_i) = q-f$ (which is the smallest amount of correct processes a quorum can have), there can be at most $k_{max} = \lfloor(\frac{n-f}{q-f})\rfloor$ such quorums $Q'_i$, therefore, at most $k_{max}$ processes in $C'$.
\end{proof}

\begin{example}
\label{ex:byz_inc}
A classical \emph{Byzantine quorum system} (BQS) uses quorums of size $q=2n/3+1$. This system is consistent as long as the faulty set has $f<n/3$ processes, but it becomes inconsistent otherwise.
As Theorem~\ref{th:classic_inc} implies, the vulnerability of the system grows with the number of faulty processes.
To illustrate this, Table~\ref{tab:tabBQS} shows how the inconsistency number varies with the number of faulty processes in a distributed system with $100$ processes.
\end{example}

\begin{table}
  \caption{Inconsistency numbers for classical BQS with 100 processes}
  \label{tab:tabBQS}
  \begin{tabular}{cccccccccccc}
    \toprule
    Number of faulty processes&0-33&34-50&51-55&56-58&59-60&61&62&63&64&65&66\\
    \midrule
    Inconsistency Number&1&2&3&4&5&6&7&9&12&17&34\\
  \bottomrule
\end{tabular}
\end{table}

\section{Asset Transfer System}
\label{sec:transfer_system}

\subsection{Preliminaries}

\myparagraph{Transactions.}
A transaction $tx \in \cT$ is a tuple $tx = (s,\tau,I)$,
where $s$ is the process identifier of the \emph{issuer}, $\tau: \Pi \rightarrow \mathbb{Z}^{+}_{0}$ is the \emph{output map} and $I \subseteq \cT$ is the set of \emph{input} transactions,
$tx$ is called \emph{outgoing from $s$} and \emph{incoming to} every $p$ such that $tx.\tau(p)>0$.
Also, every transaction in $tx.I$ must be \emph{incoming to} $tx.s$, i.e., $tx' \in tx.I \Rightarrow tx'.\tau(tx.s)>0$.
We assume that every account is owned by a dedicated process in $\Pi$.

We use the function \textit{inValue}$: \cT \rightarrow \mathbb{Z}^{+}_{0}$ to denote the sum of the amount sent to $s$ by the transaction inputs,
i.e., \textit{inValue($tx$)} $= \sum_{tx' \in I}tx'.\tau(s)$.
The function \textit{outValue}$: \cT \rightarrow \mathbb{Z}^{+}_{0}$ denotes the total amount spent in a transaction,
i.e., \textit{outValue($tx$)} $= \sum_{p \in \Pi}tx.\tau(p)$.
A transaction $tx$ is \emph{valid} \textit{iff} \textit{outValue($tx$)} $>0$ and \textit{outValue($tx$)} $=$ \textit{inValue($tx$)},
in another words, $tx$ is valid \textit{iff} the amount spent by its issuer $tx.s$ is non-zero and equal the amount $tx.s$ received in the inputs.
Since $tx$'s issuer might not send the entire value of its inputs to other processes,
in order for the transaction to be valid, we allow the remaining amount to be transferred back to the issuer in its output map.

Two distinct transactions $tx$ and $tx'$ \emph{conflict} if they are issued by the same process and share some input, i.e., $(tx.s = tx'.s) \wedge (tx.I \cap tx'.I \neq \emptyset)$.
%
We assume that the total stake is initially distributed in a special transaction $tx_{init} = (\perp, \tau_{init}, \emptyset)$ from the system originator.
The total stake of the system is therefore $\sum_{p\in\Pi}\tau_{init}(p)$.

\myparagraph{Transaction histories.}
A set of transactions $T\subseteq \cT$ is called a \emph{transaction history}.
We implicitly assume from this point on that each transaction in a history is signed by its issuer. 

$T$ generates a directed graph,
where each $tx \in T$ is a node and directed edges are drawn to $tx$ from its inputs.
Let $tx, tx' \in T$, if $tx$ is reachable from $tx'$ in this graph (i.e., there is a path from $tx'$ to $tx$), we say $tx$ depends on $tx'$.
A transaction history $T$ is \emph{well-formed} \textit{iff} it satisfies:

\begin{itemize}
    \item (T-Validity) $tx_{init} \in T \wedge \forall tx \in T, tx \neq tx_{init}: tx$ is valid;
    \item (Completeness) $\forall tx \in T, \forall tx' \in tx.I: tx' \in T$;
    \item (No-Conflict) $\forall tx, tx' \in T:$ $tx$ and $tx'$ do not conflict;
    \item (Cycle-Freedom) $\forall tx,tx' \in T:$ $tx$ depends on $tx'$ $\Rightarrow$ $tx \not\in tx'.I$.
\end{itemize}

We only consider well-formed histories from this point on.
The function \textit{balance$_{_T}: \Pi \rightarrow \mathbb{Z}$} applied to a transaction history $T$ determines the \emph{balance} of each process $w$ according to $T$:
\textit{balance$_{_T}$($w$)} is the difference between the sum of 
transfers to $w$ and the sum of transfers issued by $w$,
i.e., \[\textit{balance$_{_T}$($w$) $= \sum_{tx \in T}tx.\tau(w) \ - \sum_{tx \in T, tx.s = w}$\textit{outValue($tx$)}}\]

\begin{theorem}
Given a well-formed history $T$, for every process $w$, balance$_{_T}(w)$ $\geq 0$.
\end{theorem}

\begin{proof}
Let $\sum_{tx \in T}tx.\tau(w)$ be the \emph{incoming stake} to $w$ and let $\sum_{tx \in T}$\textit{outValue($tx$)}, with $tx.s = w$, be the \emph{outgoing stake} from $w$.
Assume that balance$_{_T}$($w$) $< 0$, then the outgoing stake is greater than the incoming stake.
The initial transaction $tx_{init}$ may only send funds to $w$, and since every other transaction $tx \in T$ is valid, $tx$ must include inputs with enough funds to cover \textit{outValue($tx$)}.
From \emph{Completeness}, for every transaction $tx$ appearing in the sum of the outgoing stake, its inputs $tx' \in tx.I$ also appear in the sum of the incoming stake. 
Therefore, the only remaining way $w$ can spend more stake than it received is to use an input more than once, which is prevented by \emph{No-Conflict}.
\end{proof}
%

Although a well-formed history has no conflicting transactions, there may exist conflicts among distinct well-formed histories.
Given a collection of well-formed transaction histories $\Gamma$, a process $r$, and $tx$ an incoming transaction to $r$.
Let $I_{tx}^{r} \subseteq \mathcal{T}$ be the set of outgoing transactions from $r$,
each $tx' \in I_{tx}^{r}$ including $tx$ in its input and appearing in some $T_{i} \in \Gamma$,
\[
I_{tx}^{r}  = \{tx'\  |\  \exists T_{i} \in \Gamma: (tx' \in T_{i}) \wedge (tx \in tx'.I) \wedge (tx'.s = r)\}.
\]

Let $|I_{tx}^{r}|=k$, we say that process $r$ \textit{$k$-spends} $tx$ in $\Gamma$.
In other words, a process $k$-spends if it issued $k$ distinct transactions appearing in $\Gamma$ using the same input.
%

\begin{definition} [Spending Number]
Let $\Gamma$ be a collection of well-formed histories.
The \emph{spending number} of $\Gamma$, noted $\gamma(\Gamma)$, is the highest amount of times an input is spent by the same process in $\Gamma$. Formally,
\begin{center}
$\gamma(\Gamma) = max(\{|I_{tx}^{r}|\ |\ \forall r \in \Pi, \forall tx$ incoming to $r\})$.
\end{center}
\end{definition}

Note that, by definition, the spending number of $\Gamma$ cannot exceed $|\Gamma|$.

\subsection{Problem Statement}

%
Every process $p \in \Pi$ maintains a \emph{local history} $T_p$.
We say that $p$ \emph{accepts} $tx$ when it adds $tx$ to $T_p$.

Ideally, we want local histories of correct processes to eventually converge.  
But this may not always be possible, as our specification allows for multiple spending: correct process may accept conflicting transactions. 
Therefore, we also introduce an accountability mechanism, expressed in the form of accusation histories. 

%
%

%
Formally, an \emph{accusation} is a tuple $(AC,P)$ consisting of a set of processes $AC \subseteq \Pi$ and a \textit{proof of misbehavior} $P$ for every process in $AC$.
$(AC,P)$ can be independently verified by a third party through the function \textit{verify-acc: }$(2^{\Pi} \times \mathcal{P}) \rightarrow \{true,false\}$.
Technically, for each process $p\in AC$, the proof $P$ must contain a set of conflicting transactions $tx_{1},...,tx_{\ell}$ signed by $p$.
We say that the accusation $(AC,P)$ \emph{refers} to $tx_{1},...,tx_{\ell}$.

Every process $p$ is also expected to maintain a local accusation history $A_p$, where each element in $A_{p}$ is an accusation tuple.
The \emph{$k$-spending asset transfer abstraction} receives inputs of the form  \emph{transfer($tx$)} and produces updates to the local histories $T_p$ and $A_p$.

Consider a run of a \emph{$k$-spending asset transfer protocol} ($k$-SAT) in a trust model $(\cQ,\cF)$ with a fixed faulty set $F\in\cF$.
Let $T_{p}(t)$ and $A_{p}(t)$ denote the transaction history and accusation history of process $p$ at time $t$, respectively. %
Let $\Gamma(t)$ denote the collection of local histories of correct processes at time $t$.
Then the run must satisfy:

\begin{description}
    \item[Validity] If a correct process issues a transaction $tx$, then every live correct process $p$ eventually adds $tx$ to $T_p$, or adds an accusation to $A_p$ referring to some transaction on which $tx$ depends.

    \item[$k$-Spending] For all $t\geq 0$, the spending number of $\Gamma(t)$ is bounded by $k$, i.e., $\gamma(\Gamma(t)) \leq k$.
    
    \item[Eventual Conviction] If correct processes $p$ and $q$ add conflicting transactions $tx$ to $T_p$ and $tx'$ to $T_q$ respectively, then they eventually add an accusation referring to $tx$ and an accusation referring to $tx'$ to $A_p$ and $A_q$.
    
    \item[Accuracy] For all $t\geq 0$ and $(AC,P)$ in $A_p(t)$: $\textit{verify-acc}(AC,P)= \textit{true}$. 
    Moreover, $\textit{verify-acc}(AC,P)$ returns $\textit{true}$ if and only if $AC\subseteq F$.  

    \item[Agreement] If a correct process $p$ adds an accusation $(AC,P)$ to $A_p$, then every correct process eventually adds $(AC,P)$ to its accusation history.
    
    \item[Integrity] If $tx.s$ is correct, a correct process $p$ adds $tx$ to $T_{p}$ only if $tx.s$ previously issued $tx$. 
    
    \item[Monotonicity] The accusation history of correct processes grows monotonically, i.e., for all $p$ correct and $t \leq t'$, $A_{p}(t) \subseteq A_{p}(t')$;
    
    \item[Termination] If a correct process $p$ adds a transaction $tx$ to $T_p$, then every live correct process $q$ eventually adds $tx$ to $T_q$ or an accusation referring to $tx$ (or some transaction on which $tx$ depends) to $A_q$.
\end{description}

\subsection{$k$-Spending Asset Transfer Protocol}

The pseudocodes of every process $p_i$ in our \emph{$k$-spending asset transfer protocol} is presented in Algorithms~\ref{alg:kTransfer1} and \ref{alg:kTransfer2}.
%
%
In our protocol, a process accepts a transaction only after hearing from a (local) quorum, and after all of the transaction's inputs are already accepted.
%

%
%

%
\myparagraph{Local Variables.} Variables \textit{echoes}, \textit{usedInp} and \textit{pending} are used in a broadcast stage of the algorithm.
The array \textit{echoes} stores received transactions echoed by other processes.
In \textit{usedInp}, $p_{i}$ stores all transactions it has witnessed to be used as inputs,
while in \textit{pending} it stores transactions that have received echoes from a quorum, but have not yet been added to the history.
The remaining variables are: $p_i$'s transaction history \textit{trHist}, $p_i$'s accusation history \textit{acHist}, and \textit{signedRequests}, an array with sets of tuples $(tx,\sigma)$, where $\sigma$ is a signature for $tx$ from $tx.s$.

The complete algorithm consists of three main blocks: the \emph{broadcast} block, the \emph{acceptance} block and the \emph{accountability} block.
In the following, we give a detailed description on how each block operates.

\myparagraph{Broadcasting transactions.} In order to issue a transaction, $p_{i}$ specifies a transaction $tx$ and invokes the operation \textit{transfer($tx$)} (we assume that transactions issued by correct processes are always valid).
Process $p_{i}$ then creates a signature $\sigma$ for $tx$ and sends them in a message \textit{REQ} to every process in the system.
Upon receiving \textit{REQ} with $tx$, $p_{i}$ stores the signed transaction in \textit{signedRequests}.
If none of $tx$'s inputs are in \textit{usedInp[$tx.s$]}, 
$p_{i}$ echoes the original signed request with the issuer's signature and adds the inputs of the transaction to \textit{usedInp[$tx.s$]}.
A message whose signature does not match its sender is ignored.

Each time $p_{i}$ receives a new \textit{ECHO} from $p_j$ for a transaction $tx$,
it stores the echoed transaction in \textit{echoes[$p_j$]}, afterwards, $p_{i}$ follows the same steps as when receiving a \textit{REQ} message.
When "enough" echoes are collected for the same transaction $tx$, if $tx$ is neither in \textit{pending} nor \textit{trHist}, it is added to \textit{pending}.

\myparagraph{Accepting transactions.} After going through the broadcast phase and adding $tx$ to \textit{pending}, the transaction needs to be verified in order to keep the local history consistent.
This verification is realized through the function \textit{ready($tx$)}, consisting of three conditions:

\begin{enumerate}
    \item All of $tx$'s inputs must have already been added to \textit{trHist};
    \item $tx$ must be valid;
    \item there is no other transaction in \textit{trHist} from $tx.s$ sharing an input with $tx$;
\end{enumerate}

If all three conditions are met, $p_{i}$ adds $tx$ to \textit{trHist} and removes it from $pending$.

\myparagraph{Treating Accusations.} Since $p_{i}$ keeps track of every received signed transaction $(tx,\sigma)$ in \textit{signedRequests} (either coming directly from a \textit{REQ} message or coming from an $ECHO$),
it can construct a proof of misbehavior after receiving signed conflicting transactions. 
The proof here consists of a pair $(tx,\sigma_j)$ and $(tx',\sigma'_j)$ containing distinct transactions from $p_{j}$ whose inputs intersect in a non empty set.
An accusation $(AC,P)$ is created using $p_{j}$'s identifier and the proof.
If it is a new accusation, $p_{i}$ adds $(AC,P)$ to \textit{acHist} and send it to every process in the network in an \textit{ACC} message.
The same steps are followed once a verifiable accusation tuple $(AC,P)$ is received from an \textit{ACC} message.

\begin{algorithm}
\SetAlgoLined
\BlankLine
\textbf{Local Variables:} \\
\textit{echoes $\leftarrow [\emptyset]^{N}$}\Comment*[r]{Array containing sets of received echoes}
\textit{usedInp $\leftarrow [\emptyset]^{N}$}\Comment*[r]{Array of inputs used by each process}
\textit{pending $\leftarrow \{\}$}\Comment*[r]{Set of transactions waiting to be accepted}
\textit{trHist $\leftarrow \{tx_{init}\}$}\Comment*[r]{Transaction history of $p_{i}$} 
\textit{signedRequests $\leftarrow [\emptyset]^{N}$}\Comment*[r]{An array with transactions signed by their issuer}
\textit{acHist $\leftarrow \emptyset$}\Comment*[r]{Accusation history of $p_{i}$}

\BlankLine

{\tiny1} \textbf{operation transfer($tx$)}: \\
    {\tiny2} \ \ \ \ $\sigma \leftarrow$ \textit{sign(self, $tx$)}; \\
    {\tiny3} \ \ \ \ \textit{send message [REQ, $tx$, $\sigma$] to all $p_{k} \in \Pi$}; \\
\BlankLine

{\tiny4} \textbf{upon receiving [\textit{REQ}, $tx$, $\sigma_{j}$] from $p_{j}$:} \\
    {\tiny5} \ \ \ \ \textit{\textbf{if (}verify($p_{j}$, $tx$, $\sigma_{j}$) $\wedge$ $(tx,\sigma_{j})$ $\not\in$ \textit{signedRequests[$tx.s$]}\textbf{):}} \\
    {\tiny6} \ \ \ \ \ \ \ \ \textit{signedRequests[$tx.s$] $\leftarrow$ \textit{signedRequests[$tx.s$]} $\cup$ $\{(tx,\sigma_{j})\}$}\Comment*[r]{stores signed transactions}
    {\tiny7} \ \ \ \ \ \ \ \ \textit{\textbf{if (}$tx.I$ $\cap$ usedInp[$tx.s$] $=$ $\emptyset$\textbf{):}} \\
    {\tiny8} \ \ \ \ \ \ \ \ \ \ \ \ \textit{usedInp[$tx.s$] $\leftarrow$ \textit{usedInp[$tx.s$] $\cup$ $tx.I$}}\Comment*[r]{stores inputs already used by $p_j$}
    {\tiny9} \ \ \ \ \ \ \ \ \ \ \ \ \textit{$\sigma \leftarrow$ sign(self, $tx$)}; \\
    {\tiny10} \ \ \ \ \ \ \ \ \ \ \ \ \textit{send message [ECHO, $(tx,\sigma_{j})$, $\sigma$] to all $p_{k} \in \Pi$}; \\

\BlankLine

{\tiny11} \textbf{upon receiving [\textit{ECHO}, }\textit{$(tx,\sigma_s)$, }\textbf{$\sigma_{j}$] from $p_{j}$:} \\
    {\tiny12} \ \ \ \ \textit{\textbf{if (}verify($p_j$, $tx$, $\sigma_{j}$) $\wedge$ verify($tx.s$, $tx$, $\sigma_s$)\textbf{):}} \\
    {\tiny13} \ \ \ \ \ \ \ \ \textit{echoes[$p_{j}$] $\leftarrow$ \textit{echoes[$p_{j}$]} $\cup$ $\{tx\}$}\Comment*[r]{stores echoed $tx$ and repeat REQ steps}
    {\tiny14} \ \ \ \ \ \ \ \ \textit{\textbf{if (}$(tx,\sigma_{s})$ $\not\in$ signedRequests[$tx.s$]\textbf{):}} \\
    {\tiny15} \ \ \ \ \ \ \ \ \ \ \ \ \textit{signedRequests[$tx.s$] $\leftarrow$ \textit{signedRequests[$tx.s$]} $\cup$ $\{(tx,\sigma_{s})\}$;} \\
    {\tiny16} \ \ \ \ \ \ \ \ \textit{\textbf{if (}$tx.I$ $\cap$ usedInp[$tx.s$] $=$ $\emptyset$\textbf{):}} \\
    {\tiny17} \ \ \ \ \ \ \ \ \ \ \ \ \textit{usedInp[$tx.s$] $\leftarrow$ \textit{usedInp[$tx.s$] $\cup$ $tx.I$}}; \\
    {\tiny18} \ \ \ \ \ \ \ \ \ \ \ \ \textit{$\sigma \leftarrow$ sign(self, $tx$)}; \\
    {\tiny19} \ \ \ \ \ \ \ \ \ \ \ \ \textit{send message [ECHO, $(tx,\sigma_{s})$, $\sigma$] to all $p_{k} \in \Pi$}; \\
\BlankLine

{\tiny20} \textbf{upon receiving echoes for $tx$ from every $q \in Q_{i}, Q_{i} \in \cQ(p_{i})$:} \\
    {\tiny21} \ \ \ \ \textit{\textbf{if (}$tx \not\in$ trHist $\wedge$ $tx \not\in$ pending\textbf{):}}\\
    {\tiny22} \ \ \ \ \ \ \ \ \textit{pending $\leftarrow$ pending $\cup$ $\{tx\}$}; \\
\BlankLine

\caption{$k$-Spending Asset Transfer System: code for process $p_{i}$ part 1}
\label{alg:kTransfer1}
\end{algorithm}

\begin{algorithm}
\SetAlgoLined
\BlankLine

{\tiny23} \textbf{upon existing $tx \in$} \textit{pending} \textbf{such that} \textit{ready($tx$)} \textbf{is true:} \\
    {\tiny24} \ \ \ \ \textit{trHist $\leftarrow$ trHist $\cup$ $\{tx\}$}\Comment*[r]{adds transaction to local history}
    {\tiny25} \ \ \ \ \textit{pending $\leftarrow$ pending$/\{tx\}$}; \\
\BlankLine

{\tiny26} \textbf{upon existing signed} $tx$ \textbf{and} $tx'$ \textbf{in} \textit{signedRequests[$p_{j}$]} \textbf{such that} $tx.I \cap tx'.I \neq \emptyset$  \textbf{:} \\
    {\tiny27} \ \ \ \ $ev1 \leftarrow$ \textit{$(tx,\sigma_{j})$}\Comment*[r]{$p_i$ witnessed conflicting transactions}
    {\tiny28} \ \ \ \ $ev2 \leftarrow$ \textit{$(tx',\sigma'_{j})$}; \\
    {\tiny29} \ \ \ \ \textit{accusation $\leftarrow (\{p_{j}\},\{ev1,ev2\})$}\Comment*[r]{$AC = \{p_j\}$, $P = \{ev1,ev2\}$}
    {\tiny30} \ \ \ \ \textit{\textbf{if (}accusation $\not\in$ acHist\textbf{):}} \\
    {\tiny31} \ \ \ \ \ \ \ \ \textit{acHist $\leftarrow$ acHist $\cup$ $\{$accusation$\}$}\Comment*[r]{adds accusation to local history}
    {\tiny32} \ \ \ \ \ \ \ \ \textit{send message [\textit{ACC}, accusation] to all $p_{k} \in \Pi$}; \\
\BlankLine

{\tiny33} \textbf{upon receiving a message [\textit{ACC}, }\textit{accusation}\textbf{] from $p_{j}$:} \\
    {\tiny34} \ \ \ \ \textit{\textbf{if (}verify-acc(accusation) $\wedge$ accusation $\not\in$ acHist\textbf{):}} \\
    {\tiny35} \ \ \ \ \ \ \ \ \textit{acHist $\leftarrow$ acHist $\cup$ $\{$accusation$\}$}; \\
    {\tiny36} \ \ \ \ \ \ \ \ \textit{send message [\textit{ACC}, accusation] to all $p_{k} \in \Pi$}; \\
\BlankLine

{\tiny37} \textbf{function} \textit{ready($tx$)}\textbf{:} \\
{\tiny38} \ \ \ \ \textit{c1 $\leftarrow$ $\forall tx' \in tx.I: tx' \in$ trHist}\Comment*[r]{\emph{Completeness}}
{\tiny39} \ \ \ \ \textit{c2 $\leftarrow$ TRUE iff $tx$ is valid}\Comment*[r]{\emph{T-Validity}}
{\tiny40} \ \ \ \ \textit{c3 $\leftarrow$ $\forall tx' \in$ trHist $: (tx.s = tx'.s) \Rightarrow (tx.I \cap tx'.I = \emptyset)$}\Comment*[r]{\emph{No-Conflict}}
{\tiny41} \ \ \ \ \textit{return $c1 \wedge c2 \wedge c3$}; \\
\BlankLine

\caption{$k$-Spending Asset Transfer System: code for process $p_{i}$ part 2}
\label{alg:kTransfer2}
\end{algorithm}

\myparagraph{Correctness.}
Consider executions of Algorithms~\ref{alg:kTransfer1} and~\ref{alg:kTransfer2} assuming trust model $(\cQ,\cF)$ with inconsistency number 
$k_{max}$.
Let $F\in\cF$ be the corresponding faulty set.

\begin{lemma}
\label{lm:well_formation}
The history $T_p$ of a correct process $p$ is well-formed.
\end{lemma}

\begin{proof}
The default value for \textit{trHist} is $\{tx_{init}\}$, which is well-formed by definition.
Now assume that at some point \textit{trHist} is well-formed.
Before adding a new transaction $tx$ to \textit{trHist}, $p$ previously uses \textit{ready($tx$)} to check wheter $tx$ is valid and that the resulting history satisfies \textit{No-Conflict} and \textit{Completeness} (lines 37 to 41).
By construction, \textit{trHist} is also \emph{Cycle-Free}: suppose $\{tx\}$ $\cup$ \textit{trHist} creates a cycle, that is, $\exists tx' \in \{tx\}$ $\cup$ \textit{trHist} on which $tx$ depends where $tx \in tx'.I$.
This is clearly not possible: since \textit{trHist} is well-formed, $\forall tx'' \in tx'.I: tx'' \in$ \textit{trHist}, but $tx \not\in \textit{trHist}$, a contradiction.
\end{proof}

\begin{lemma}[\emph{$k$-Spending}]
\label{lm:k_spending}
 At any time $t$, the spending number of $\Gamma(t)$ is bounded by $k_{max}$, i.e., $\gamma(\Gamma(t)) \leq k_{max}$.
\end{lemma}

\begin{proof}
Let $r \in F$ and $tx$ an incoming transaction to $r$.
Suppose $r$ spends $tx$ $k$ times in $\Gamma(t)$, with $k > k_{max}$.
We assume, without loss of generality, that $r$ is the process that multiple spent the maximal number of times in $\Gamma(t)$, which means $\gamma(\Gamma(t))=k$.
We can make the following observations about the algorithm:

\begin{enumerate}
    \item A correct process $p$ adds a transaction $tx'$ to its history only if it received \textit{ECHO} messages for $tx'$ from every process in a quorum $Q \in \cQ(p)$ (guard in line 20).
    \item A correct process $p$ checks if any input of a received transaction is already in \textit{usedInp} before echoing it (lines $7$ and $16$), and if it sends \textit{ECHO} for a transaction, $p$ adds all of its inputs to \textit{usedInp} (lines $8$ and $17$). Therefore, $p$ can send \textit{ECHO} for at most one transaction from $r$ which has $tx$ as an input.
\end{enumerate}

Let $p_i$ and $p_j$ correct accept conflicting $tx_i$ and $tx_j$ from $r$ after receiving echoes from $Q_i \in \cQ(p_i)$ and $Q_j \in \cQ(p_j)$ respectively.
From (2) above, we conclude that $Q_i \cap Q_j \subseteq F$,
otherwise a correct process in the intersection would have echoed two different transactions sharing some input(s) from $r$, which is not allowed by the algorithm (the guards in lines $7$ and $16$ prevent this).

Since $r$ $k$-spends $tx$ in $\Gamma(t)$,
there exists $p_1,...,p_k$ correct that accepted, respectively, conflicting $tx'_1,...,tx'_k$ from $r$ using $tx$ as input.
Now let $Q_1 \in \cQ(p_1),...,Q_k \in \cQ(p_k)$ be the quorums each process received echoes from before adding the conflicting transaction to its history.
Then we can construct a quorum map $S$ satisfying $S(p_i) = Q_i$ for $i= 1,...,k$,
and a graph $G_{F,S} \in \cG_{\cQ,\cF}$ of which $C = \{p_1,...,p_k\}$ is an independent set,
since from (1) and (2) above: $\forall p_i,p_j \in C, i \neq j: S(p_i) \cap S(p_j) \subseteq F$.
However, $k_{max}$ is the inconsistency number of $(\cQ,\cF)$, meaning that there cannot be a graph $G_{F,S} \in \cG_{\cQ,\cF}$ with an independent set of size $k > k_{max}$, a contradiction.
\end{proof}

\begin{lemma}[\emph{Eventual Conviction}]
\label{lm:conviction}
If correct processes $p$ and $q$ add conflicting transactions $tx$ to $T_p$ and $tx'$ to $T_q$ respectively, then they eventually add an accusation referring to $tx$ and an accusation referring to $tx'$ to $A_p$ and $A_q$.
\end{lemma}

\begin{proof}
Let correct processes $p$ and $q$ accept conflicting transactions $tx$ and $tx'$ respectively.
They have previously received echoes for $tx$ (in $p$'s case) and $tx'$ (in $q$'s case), storing the original signed requests in their local \textit{signedRequests} (lines $14$ and $15$).
There are two scenarios to consider for each one of them (we describe it here only for $p$ for simplicity):
$p$ echoed $tx$ before adding it to $T_p$, or $p$ did not echo $tx$.
If $p$ echoed $tx$, then $q$ will eventually receive the echo with a signed request for $tx$ from $p$,
which allows $q$ to construct and relay an accusation for $tx.s$ (in lines $26$ to $32$) using this request together with the one for $tx'$ already stored in $q$'s \textit{signedRequests}
(e.g. assigning the request for $tx$ to \textit{ev1} in line $27$ and the request for $tx'$ to \textit{ev2} in line $28$).
Eventually $p$ will receive an \textit{ACC} message from $q$ containing this accusation and will add it to its accusation history.

Now if $p$ did not echo $tx$, then it must have echoed for another conflicting transaction $tx''$, which means $p$ can construct an accusation using the respective signed requests for $tx$ and $tx''$ as described above.
This accusation is sent to every process in the network and is eventually received by $q$, which adds it to its accusation history.
These scenarios occur in the same way for $q$ and $tx'$.
Ultimately, both $p$ and $q$ end up adding accusations referring to $tx$ and $tx$' to their histories.
\end{proof}

\begin{lemma}[\emph{Termination}]
\label{lm:termination}
If a correct process $p$ adds a transaction $tx$ to $T_p$, then every live correct process $q$ eventually adds $tx$ to $T_q$ or an accusation referring to $tx$ (or some transaction on which $tx$ depends) to $A_q$.
\end{lemma}

\begin{proof}
Recall that a process is live if it has a quorum composed of only correct processes.

We first show the following: If a correct process adds a transaction $tx$ to its local \textit{pending} set,
then every live correct process eventually does so or adds an accusation referring to $tx$ to its local \textit{acHist}.

Let $p$ be a correct process that adds $tx$ to its \textit{pending} after receiving \textit{ECHO} messages for $tx$ from a quorum.
There are two cases to consider, depending on whether $p$ previously echoed $tx$ or not.

If $p$ did not echo $tx$, then it echoed a conflicting $tx'$ and is able to build an accusation $(AC,P)$ with the original requests for $tx$ and $tx'$ (lines $26$ to $29$).
Then, $p$ adds the accusation to its \textit{acHist} and sends $(AC,P)$ to all processes.
Every correct process $q$ eventually receives the accusation and also adds it to \textit{acHist}.

Suppose now that $p$ echoed $tx$.
If no process sent \textit{ECHO} or \textit{REQ} for a conflicting transaction, then every correct process eventually receives and echoes $tx$.
If a correct process $q$ is live, it will eventually receive enough echoes and add $tx$ to pending.
On the other hand, if a process in $q$'s live quorum had echoed a conflicting transaction, $q$ will receive the conflicting requests, build an accusation $(AC,P)$ referring to $tx$ and $tx'$ and send it to all processes.
Then, as described previously, every correct process eventually adds $(AC,P)$ to \textit{acHist}.

Now suppose $p$ also adds $tx$ to its \textit{trHist}.
We make the following observations about the algorithm:
before being added to \textit{trHist},
any transaction $tx'$ is first added to \textit{pending}  (guard in line $23$).
Also, by Lemma~\ref{lm:well_formation}, every transaction on which $tx'$ depends must have been previously added to \textit{trHist}.
Let \textit{deps($tx$)} include $tx$ and every transaction on which $tx$ depends.
It follows that $p$ previously added every $tx' \in$ \textit{deps($tx$)} to \textit{pending}.
The following three cases are then possible for a live correct process $q$:

\begin{enumerate}
    \item $q$ eventually adds every $tx' \in$ \textit{deps($tx$)} to its \textit{pending}.
    If no transaction in \textit{trHist} conflicts with them, $q$ adds every such $tx'$ to \textit{trHist}.
    
    \item $q$ has already added a transaction to \textit{trHist} that conflicts with some $tx' \in$ \textit{deps($tx$)}.
    In this case, it received conflicting requests. 
    $q$ will then build and send everybody an accusation including the signed requests for the respective transactions.
    \item $q$ never adds one (or more) $tx' \in$ \textit{deps($tx$)} to \textit{pending}, in which case, as previously shown, $q$ eventually adds an accusation referring to $tx'$ to \textit{acHist}.
\end{enumerate}

Therefore, if a correct process $p$ adds a transaction $tx$ to its \textit{trHist} and a live correct process $q$ is never able to do so, then $q$ eventually adds an accusation to its \textit{acHist} referring to $tx$ or some transaction on which $tx$ depends.
\end{proof}

\begin{lemma}[\emph{Validity}]
\label{lm:validity}
If a correct process issues a transaction $tx$, then every live correct process $p$ eventually adds $tx$ to $T_p$, or adds an accusation to $A_p$ referring to some transaction on which $tx$ depends.
\end{lemma}

\begin{proof}
If a correct process $p$ sends a \textit{REQ} message for $tx$, eventually every correct process echoes $tx$ and every live correct process adds $tx$ to its \textit{pending}.
Since $p$ is correct, it will not send conflicting requests, thus no accusation referring to $tx$ can be produced.
Also, $p$ must have previously added every transaction on which $tx$ depends to \textit{trHist},
which from Lemma~\ref{lm:termination}, if a live correct process $q$ does not add said transactions to its \textit{trHist} (and consequently $tx$, since it is in \textit{pending}),
then $q$ eventually adds an accusation to \textit{acHist} referring to some transaction on which $tx$ depends.
\end{proof}

\begin{theorem}
\label{th:k_spending}
Consider the trust model $(\cQ,\cF)$ with inconsistency number $k_{max}$. 
Then Algorithms~\ref{alg:kTransfer1} and~\ref{alg:kTransfer2} implement $k_{max}$-spending asset transfer abstraction.
\end{theorem}

\begin{proof}
Lemma~\ref{lm:well_formation} shows that correct processes always maintains well-formed local transaction histories.
The \emph{$k_{max}$-Spending}, \emph{Eventual Conviction}, \emph{Termination} and \emph{Validity} properties are shown in Lemmas~\ref{lm:k_spending} to~\ref{lm:validity}.

\emph{Accuracy}, \emph{Monotonicity}, \emph{Agreement} and \emph{Integrity} are immediate.
A correct process adds an accusation $(AC,P)$ to its \textit{acHist} only if it can verify that messages for conflicting transactions in $P$ were indeed signed by the process in $AC$ (\emph{Accuracy}).
The set \textit{acHist} may only grow with time (\emph{Monotonicity}).
Moreover, once a correct process adds an accusation to its \textit{acHist}, it then sends the accusation to every other process in an \textit{ACC} message.
This message is eventually received by every correct process, which verifies and adds the accusation to its \textit{acHist} (\emph{Agreement}).
Finally, the signature of a correct process for a transaction request cannot be forged (\emph{Integrity}).

\end{proof}

\section{Relaxed Broadcast Abstraction and Lower Bounds}
\label{sec:lower_bounds}

In this section, we show that the inconsistency number of $(\cQ,\cF)$ is optimal for $k$-SAT, by relating the problem to the fundamental \emph{broadcast abstraction}.
%
%
%
The abstraction exports one operation \textit{broadcast($m$)} and enables a callback \textit{deliver($m$)}, for $m$ in a value set $\mathcal{M}$. %
We assume that each broadcast instance has a dedicated \emph{source}, i.e., the process invoking the broadcast operation.
The following abstraction and lower bounds for the relaxed broadcast  were first introduced in a previous work~\cite{bezerra2021relaxed}, we also present them here for completeness.

%
We now describe \emph{$k$-Consistent Broadcast} ($k$-CB).
Given a trust model $(\cQ,\cF)$, in every execution with a fixed $F \in \cF$, a $k$-Consistent Broadcast protocol ensures the following properties:

\begin{itemize}
    \item (Validity) If the source is correct and broadcasts $m$, then every \emph{live} correct process eventually delivers $m$.
    \item ($k$-Consistency) Let $M$ be the set of values delivered by the correct processes, then $|M| \leq k$.
    \item (Integrity) A correct process delivers at most one value and, if the source $p$ is correct, only if $p$ previously broadcast it.
\end{itemize}

This protocol is a generalized version of an abstraction known as \emph{Consistent Broadcast}~\cite{cachin2011introduction}.
\emph{Validity} in Consistent Broadcast guarantees that a broadcast value is delivered by every correct process.
Also, correct processes cannot deliver different values.
Note that if every correct process is live and $k=1$, then $k$-CB implements Consistent Broadcast.

\subsection{Lower bound for $k$-Consistent Broadcast}

We restrict our attention to \emph{quorum-based protocols},
initially introduced in the context of consensus algorithms~\cite{losa2019stellar}.
Intuitively, in a quorum-based protocol, a process $p$ should be able to make progress if the members in one of its quorums $Q \in \cQ$ appear correct to $p$.
This should hold even if the actual set of correct processes in this execution is different from $Q$. 
Formally, we make the following assumption about algorithms implementing $k$-CB:

\begin{itemize}
    \item (Local Progress) For all $p \in \Pi$ and $Q \in \cQ(p)$, 
    there is an execution in which only the source and processes in $Q$ take steps, $p$ is correct, and $p$ delivers a value.
\end{itemize}

Consider a trust model $(\cQ,\cF)$ and
its graph representation (Section~\ref{sec:dc_trust}). Let $k_{max}$ be its inconsistency number.
%
%

\begin{theorem}
\label{th:allbound}
%
%
No algorithm can implement $k$-CB in a trust model $(\cQ,\cF)$ such that $k < k_{max}$.
\end{theorem}

\begin{proof}
Let $G_{F,S}$ be the graph generated over fixed $F \in \cF$ and $S \in \cS$ and $C = \{p_1,...,p_k\}$ an independent set in $G_{F,S}$ of size $k$.
We proceed to show that there exists an execution where $k$ different values are delivered by processes in $C$.
Let $r \in F$ be the source, by the definition of \emph{Local Progress}, it exists an execution $\Sigma_i$ for each $p_{i}$ where $part(\Sigma_i)= \{r\} \cup S(p_i)$, in which $p_i$ delivers a value $m_i$.
In this case, $r$ and other faulty processes in $S(p_i)$ appear correct within $\Sigma_i$.
Since $\forall p_i,p_j \in C: S(p_i) \cap S(p_j) \subseteq F$, we can build $\Sigma$ such that all executions $\Sigma_i$ are subsequences of $\Sigma$, in which no correct process receives any information of conflicting values before $p_{1},...,p_{k}$ deliver $m_{1},...,m_{k}$, respectively.

Now let $G_{F,S}' \in \cG_{\cQ,\cF}$ be a graph whose independence number is $k_{max}$,
there exists an independent set $C_{max}$ of size $k_{max}$ in $G_{F,S}'$.
As shown above, it is always possible to build an execution where $k_{max}$ processes deliver $k_{max}$ distinct values before any correct process is able to identify the misbehavior.
\end{proof}

Intuitively, if two correct processes have quorums that do not have a correct process in the intersection,
they might deliver distinct values before noticing any misbehavior in the execution.
Within an independent set, the quorums of every pair of nodes do not intersect in a correct process,
and $k_{max}$ represents the highest possible independent set in $\cG_{\cQ,\cF}$,
thus establishing the lower bound for $k$-CB.

\subsection{Relating $k$-Spending Asset Transfer and $k$-Consistent Broadcast}

We show now that having a protocol implementing $k$-SAT, one implement $k$-CB, which implies that the lower bound established in Theorem~\ref{th:allbound} also holds for $k$-SAT.

\begin{theorem}
\label{th:red}
$k$-SAT can be used to implement $k$-CB. 
\end{theorem}
\begin{proof}
Suppose that we have a protocol implementing $k$-SAT.
We show how one can slightly modify the $k$-SAT protocol to implement $k$-CB with a message set $\cM$.
First, we amend the notion of a transaction by allowing the source to attach a message $m\in \cM$ to it.
Also we let $tx_{init}$ assign some funds to the source $p$ and let $\tau$ be some matching output map.

Therefore, to broadcast a message $m$, $p$ issues the transaction $tx=(p,\tau,\{tx_{init}\},m)$.
Whenever a correct process $q$ adds $tx=(p,\tau,\{tx_{init}\},m')$ to $T_q$, it issues  \textit{deliver($m$)}.
If $p$ is correct, every live correct process eventually delivers it, that is $k$-SAT \emph{Validity} implies $k$-CB \emph{Validity}.
The \emph{$k$-Spending} property implies that up to $k$ conflicting transactions issued by $p$  with $tx_{init}$ as input can be accepted by correct processes.
Thus, at most $k$ distinct messages might be "delivered", which implies \emph{$k$-Consistency}.
Trivially, $k$-SAT \emph{Integrity} implies $k$-CB \emph{Integrity}.
\end{proof}

Theorems~\ref{th:allbound} and~\ref{th:red} imply that
Algorithms~\ref{alg:kTransfer1} and~\ref{alg:kTransfer2} implements the $k$-SAT abstraction with optimal $k$.

\section{Achieving Local Consistency}
\label{sec:clusters}

In Section~\ref{sec:transfer_system}, we show how to limit the actions of a malicious user given the trust model $(\cQ,\cF)$: the number of times an input can be spent is bounded by the inconsistency number of $(\cQ,\cF)$.
However, one may also expect the trust assumptions to yield groups of participants with good connections, or \emph{clusters}.
Intuitively, processes within a cluster maintain a mutually consistent view of the system state, that is, their histories do not contain conflicting transactions.
In the following, we formally define the notion of clusters.

We expand the transaction tuple to add a \emph{timestamp} indicator $tm$: $tx = (s,\tau,I,tm)$.
A transaction $tx$ \emph{precedes} $tx'$ \textit{iff} $(tx.s = tx'.s) \wedge (tx.tm = tx'.tm-1)$.
From now on, we also require each well formed history $T$ to include a predecessor for every transaction in it\footnote{Except for transactions with timestamp $1$. Moreover, correct processes are expected to issue a single transaction per timestamp.}.
Let $T.r$ be the projection of a history $T$ to transactions issued by process $r$, i.e., $tx \in T.r \Leftrightarrow (tx \in T \wedge tx.s = r)$.

\myparagraph{Clusters.} A subset of histories $C \subseteq \Gamma$ is a \emph{cluster} in $\Gamma$ if it satisfies:
\[ \forall r \in \Pi, \forall T_i, T_j \in C: (T_i.r \subseteq T_j.r) \vee (T_j.r \subseteq T_i.r) \]

Now let $\mathbb{C} = [C_1,...,C_k]$ be an array of non-empty clusters such that $\bigcup \mathbb{C} = \Gamma$.
We say that $\mathbb{C}$ is \emph{minimum} \textit{iff} for every such array of clusters $\mathbb{C}'$: $|\mathbb{C}| \leq |\mathbb{C'}|$,
that is, $\mathbb{C}$ covers the entirety of $\Gamma$ with the least amount of clusters.
We then call $|\mathbb{C}|$ the \emph{cover number} of $\Gamma$.
Note that this number is at least $\gamma(\Gamma)$, since no pair of transaction histories with conflicting transactions can belong to the same cluster.

The following question remains: in a run of a relaxed asset transfer system, what is the bound on the number of clusters necessary to cover the histories of correct processes?

Consider a run of $k$-spending asset transfer with fixed $F \in \cF$, and let $\Gamma(t)$ be the collection of histories of correct processes.
Formally, we want to find the minimum possible $\kappa$ such that:

\begin{description}
    \item[$\kappa$-Views] For all $t\geq 0$, the cover number of $\Gamma(t)$ is bounded by $\kappa$.
\end{description}

By definition, histories within a cluster $C$ do not contain conflicting transactions,
therefore $\bigcup_{C}$ results in a well-formed history.
Intuitively, the system works as if processes whose history belongs to $C$ are following the same consistent \emph{view} of the system state.

We expect $\kappa$ to be higher than the spending number of $\Gamma$.
For instance, in a run of Algorithms~\ref{alg:kTransfer1} and \ref{alg:kTransfer2}, distinct malicious issuers might exploit vulnerabilities of different trust graphs in order to make the histories of more than $k_{max}$ correct processes conflict.
In addition, for the lower bound in Section~\ref{sec:lower_bounds}, we relate asset transfer to a single-instance broadcast in which only one source broadcasts a value.
In order to know in how many views $\Gamma$ can be partitioned, it is required a deeper consistency analysis of a long-lived abstraction (with multiple potential sources).
Therefore, the question of what is the smallest $\kappa$ that can be implemented remains open.

\section{Related Work}
\label{sec:related_work}
Damg{\aa}rd et al.~\cite{damgaard2007secure} appear to be the first to consider the decentralized trust setting. 
They introduced the notion of \textit{aggregate adversary structure}~$\mathcal{A}$: each node is assigned with a collection of subsets of nodes that the adversary might corrupt at once. 
In this model, assuming \emph{synchronous} communication, 
they discuss solutions for broadcast, verifiable secret sharing and multiparty computation.

Ripple~\cite{schwartz2014ripple} and Stellar~\cite{mazieres2015stellar}, conceived as \emph{open} payment systems, use decentralized trust as an alternative to \emph{proof-of-work}-based protocols~\cite{nakamoto2008bitcoin,ethereum}.
In the Ripple protocol, each participant expresses its trust assumptions in the form of a \textit{unique node list} (UNL), a subset of system members. 
To accept a transaction, a node needs acknowledgement from a set of at least $80\%$ of its UNL (which can be seen as a quorum). 
%
%
The Ripple white paper~\cite{schwartz2014ripple} assumes that up to $20\%$ of members in an UNL are Byzantine, stating that an overlap of at least $20\%$ between every pair of UNLs in enough to prevent \emph{forks}. 
Later analyses suggest this overlap to be more than $40\%$~\cite{armknecht2015ripple} without Byzantine faults, and more than $90\%$ with the same original assumptions~\cite{chase2018analysis} (up to $20\%$ Byzantine members in an UNL). 
Chase and MacBrough~\cite{chase2018analysis} also provide an example in which liveness of the protocol is violated even with $99\%$ of overlap.

Stellar consensus protocol~\cite{mazieres2015stellar} uses a \emph{Federated Byzantine Quorum System} (FBQS). 
A quorum $Q$ in the FQBS is a set that includes a \emph{quorum slice} (a trusted subset of members) for every member in $Q$. 
Correctness of Stellar depends on the individual trust assumptions, and stronger properties are guaranteed for nodes trusting the "right guys", which are in so called \textit{intact sets}.
Garc\'ia-P\'erez and Gotsman~\cite{garcia2018federated} treated Stellar consensus formally, by relating it to Bracha's Broadcast Protocol~\cite{bracha1987asynchronous}, build on top of a FBQS.
Their analysis has been later extended~\cite{garcia2019deconstructing} to a variant of state-machine replication protocol that allows \emph{forks}, where 
disjoint intact sets may maintain different copies of the system state.

%
%
%

Recent works~\cite{losa2019stellar,cachin2020asymmetric} investigate more general formalizations of decentralized trust. 
%
%
Cachin and Tackmann~\cite{cachin2020asymmetric} model trust assumptions via an \emph{asymmetric fail-prone system}, an array $[\cF_1,...,\cF_n]$ of adversary structures (or fault models), where each $\cF_i$ is chosen by $p_i$ as its local fault model.
For this model, they devise an \emph{asymmetric Byzantine quorum system} (ABQS), an array of quorum systems $[\cQ_1,...,\cQ_n]$ that should satisfy specific intersection and availability properties, so that certain problems, such as broadcast and storage, can be solved.
%

%
%

%
%
%
%

%
Losa et al.~\cite{losa2019stellar} introduced the notion of a \textit{Personal Byzantine Quorum System} (PBQS), where every process chooses its quorums with the restriction that if $Q$ is a quorum for a process $p$, then $Q$ includes a quorum for every process $q' \in Q$.
The PBQS model is then discussed in relations to Stellar consensus~\cite{mazieres2015stellar}.   
%
More precisely, they characterize the conditions on PBQS under which a \emph{quorum-based} protocol (captured by our Local Progress condition) ensures that a well-defined subset of processes (a \emph{consensus cluster}) can maintain safety and liveness of consensus.

%
In contrast, we allow the processes to directly choose their quorums, and we address the question of what is the ``best'' consistency a cryptocurrency can achieve within this trust model.
The measure of consistency is quantified here as the spending number.
In a way, unlike this prior work on decentralized trust, instead of searching for the weakest trust model that enables solutions to a given problem, we determine the ``strongest'' problem (in a specific class) that is possible to solve in a given model.

%
%
%
%
%

%
%

%
In the context of distributed systems, accountability has been proposed as a mechanism to detect ``observable'' deviations of system nodes from the algorithms they are assigned with~\cite{detection-case,detection-problem,peerreview}.      
Recent proposals~\cite{polygraph,rala} focus on \emph{application-specific} accountability that only heads for detecting misbehavior that affects correctness of the problem to be solved, e.g., consensus~\cite{polygraph} or lattice agreement~\cite{rala}.
Our $k$-SAT algorithm generally follows this approach, except that it implements a relaxed form of asset transfer system, but detects violations that affect correctness of the stronger, conventional asset transfer abstraction~\cite{cons-crypto}.

\section{Discussion and Future Work}
\label{sec:conclusion}

\myparagraph{Generalizing the inconsistency measure.}
%
%
%
Our notion of the inconsistency number of a trust model $(\cQ,\cF)$ serves to quantify the amount of times a process can spend the same input in our cryptocurrency implementation (or the number of distinct values that can be delivered by correct processes in our broadcast abstraction).
As discussed in Section~\ref{sec:clusters}, we can expand the analysis to understand in how many clusters (each of which follows a distinct view of the system state) the set of correct participants can be split.

It is also very interesting to apply ``inconsistency metrics'' for solving other, more general problems in the decentralized trust setting.
%
%
Consider for example \emph{State Machine Replication} (SMR) protocols~\cite{paxos,pbft}.
In these protocols, correct processes agree on a global history of concurrently applied operations, and thus witness the same evolution of the system state. 
One way to relax consistency guarantees of SMR protocols in decentralized trust settings is to bound the number $k$ of diverging histories (the maximum degree of the fork).
The question is then how to relate $k$ to the trust model $(\cQ,\cF)$.

\myparagraph{Reconfiguration of inconsistent states.} 
Our $k$-SAT abstraction provides the application with the history of valid transactions and a record of misbehaving parties.
An important question is left open: 
once correct processes accept conflicting transactions and accusations against Byzantine processes are raised, what is next? Is there a way to render the system back to a consistent state?
Although there is no general answer to these questions --- it comes down to what better suits the application --- we point out some of the suitable strategies. 

A natural response to this is to \emph{reconfigure} both the trust model and the states of the processes, in order to achieve some desired level of consistency.
The immediate use of an accusation $(AC,P)$ is to rearrange the trust assumptions, removing the misbehaving parties $AC$ from the system. For example, the application might use the accusation history to suggest new (improved) quorum systems to system members.

When it comes to reconfiguration of the system states, we face a more challenging task.
%
%
%
%
Indeed, some correct processes may have already used ``compromised'' (multiply spent) assets in their transactions.  
``Merging'' conflicting histories into a consistent global state might affect the stake distribution, which can be hard to resolve without changing the application semantics.
We present two strategies
that make use of the transactions identified in accusation proofs.
The first approach, alluding to the financial system, is to let them keep (and use as input) any accepted conflicting transaction after the misbehaving parties are excluded from the system.
This can have implications on the total stake of the system: depending on how much stake was spent, the total system stake may grow.
The advantage of this approach is that the system can stay live and always give finality to operations once transactions are accepted. 
As a downside, Byzantine processes could fake incoming conflicting transactions in order to grow their own stake.
This can be handled by requiring every process that wants to use a conflicting transaction to prove (by means of collected signatures) that it has indeed received the acknowledgements from a quorum, although this would require trust assumptions to be globally known. 

The second, and probably the most straightforward approach, is to \emph{rollback} any transaction $tx$ appearing in a proof, i.e., removing $tx$ and every transaction depending on $tx$ from transaction histories.
Surely, this comes with the downside of invalidating a previously accepted transaction, which might affect correct system members in real life.
As an alternative way of compensating correct processes in this case, the application might opt to redistribute the stake from the misbehaving parties among the harmed members.

Given a strategy for the reconfiguration of system members and states, an interesting course to follow would be in self-reconfigurable systems~\cite{rala}: the protocol automatically rearrange trust assumptions and merge conflicting histories to keep the system live.

\myparagraph{Composition of trust.}
Alpos et al.~\cite{alpos2021trust} show how to compose trust models of different (possibly disjoint) systems.
%
%
Given two asymmetric fail-prone systems $[\cF_1,...,\cF_n]$ and $[\cF'_1,...,\cF'_m]$ and matching decentralized quorum systems, a \emph{composed} trust model can be constructed. 
%
In the context of our relaxed cryptocurrency protocols, it is appealing to understand how the spending number of a composition of two independent systems may depend on the spending number of its components.

\bibliographystyle{abbrv}

\bibliography{references}

\end{document}